\documentclass{ifacconf}
\usepackage[english]{babel}
\usepackage{amsmath,amsfonts}
\usepackage{amsthm}
\usepackage{comment}
\usepackage{graphicx}      
\usepackage{natbib}  
\newtheorem{theorm}{Theorem}[section]
\newtheorem{corol}{Corollary}[theorm]
\newtheorem{lemm}[theorm]{Lemma}
\newtheorem{propos}[theorm]{Proposition}
\theoremstyle{definition}
\newtheorem{defin}{Definition}

\usepackage{caption}    
\usepackage{subfig}  
\usepackage{mwe,wrapfig}
\usepackage[dvipsnames]{xcolor}

\begin{document}
\begin{frontmatter}

\title{Bifurcations in Latch-Mediated Spring Actuation (LaMSA) Systems} 

\author[First]{Vittal Srinivasan$^1$} 
\author[First]{Nak-seung Patrick Hyun$^2$}

\address[First]{ Authors are affiliated to Elmore Family School of Electrical and Computer Engineering, Purdue University, West Lafayette, IN, USA (e-mail: [1] srini133@purdue.edu, [2] nhyun@purdue.edu).}

\begin{abstract}         
: In nature, different species of smaller animals produce ultra-fast movements to aid in their locomotion or protect themselves against predators. These ultra-fast impulsive motions are possible, as often times, there exist a small latch in the organism that could hold the potential energy of the system, and once released, generate an impulsive motion. These type of systems are classified as Latch Mediated Spring Actuated (LaMSA) systems, a multi-dimensional, multi-mode hybrid system that switches between a latched and an unlatched state. The LaMSA mechanism has been studied extensively in the field of biology and is observed in a wide range of animal species, such as the mantis shrimp, grasshoppers, and trap-jaw ants. In recent years, research has been done in mathematically modeling the LaMSA behavior with physical implementations of the mechanism. A significant focus is given to mimicking the physiological behavior of the species and following an end-to-end trajectory of impulsive motion. This paper introduces a foundational analysis of the theoretical dynamics of the contact latch-based LaMSA mechanism. The authors answer the question on what makes these small-scale systems impulsive, with a focus on the intrinsic properties of the system using bifurcations. Necessary and sufficient conditions are derived for the existence of the saddle fixed points. The authors propose a mathematical explanation for mediating the latch when a saddle node exists, and the impulsive behavior after the bifurcation happens. 
\end{abstract}

\begin{keyword}
Bio-inspired Robots, Impulsive Systems, Bifurcations, Multidimensional Systems, Hybrid Systems, Mechanical Systems.
\end{keyword}
\end{frontmatter}
\section{Introduction}
\vspace{-3mm}
Latch-mediated Spring Actuation (LaMSA) systems form a class of non-linear systems that use the interplay of latches and springs to generate movement through the mediation of energy. The actuation of the spring causes a change in potential energy stored in the spring to kinetic energy that generates movement. The mediation of the latch acts as a control to facilitate this energy transfer, creating an ultra-speed movement or impulse (\cite{SL2019}). The behavior of these impulsive systems has been studied in detail through a biological lens. \cite{CP2018} investigated the highly accelerated punches of the mantis shrimp, which generates a force equivalent to a tiger's bite against other underwater predators or prey. These strikes reach peak acceleration of the order of
$10^6\: \mbox{rad s}^{-2}$, making it one of the fastest strikes in the animal kingdom. The LaMSA mechanism is also observed in insects and amphibians to aid their movement. \cite{RA2019} analyzed the nature of the elastic jump in Cuban tree frogs on different base substrates to investigate the energy flow in the frogs while jumping.

Under the broad class of LaMSA, systems can be further classified based on the type of latch into (i) contact latch, (ii) fluidic latch, and (iii) geometric latch. Contact latch mechanisms have a physical latch that holds the spring system in place. The unlatching leads to the projectile's
ballistic takeoff, as seen in the mandibles of trap-jaw ants (\cite{PB2006}). Fluidic latches are mediated by the microscopic and macroscopic fluid properties within a system, as seen in most fungal species (\cite{LP2017}). The geometric latch has a state dependence based on some geometric configurations in the system, such as forces, moment arms, the center of mass, etc., as seen in snapping shrimp (\cite{LP2023}) and other organisms. 

With this knowledge, significant progress has been made in replicating these mechanisms in bio-inspired robotics. \cite{IB2018} proposed a framework for synthesizing, scaling, and analyzing power-amplified biological systems with the dynamic coupling of motors, springs, and latches. \cite{DB2020} created a physical model of the contact latch model and investigated the idea of varying unlatch time for the LaMSA system with different latch radii. \cite{HO2023} proposed a generalized model of the contact latch and explored the idea of controlling the unlatch time by varying the latch velocity. Other classifications of the LaMSA mechanisms, such as the geometric latch in the mantis shrimp, have been explored with simulations and physical implementations by \cite{SH2021}. 

From the dynamics perspective, the LaMSA mechanism can be viewed as a hybrid system that switches between modes of latching to unlatching. We can further classify LaMSA under Multi-Dimensional Multi-Mode $(\mbox{M}^3\mbox{D})$ systems, which was introduced by \cite{EV2010} as the nature of the state space changes from a 1 DoF system when latched to a 2 DoF system when unlatched. The idea of bifurcations is used frequently in the analysis of switched dynamics. \cite{ZW2014} analyzed the Hopf and fold bifurcations in the centrifugal flywheel governor system, a fourth-order non-linear system whose torque switches between a constant load and periodic oscillations. 

The vast literature on LaMSA provides a rich background on the behavior of various animal species from a biological perspective. It brings forth a deeper understanding of how these organisms use the LaMSA mechanism, for example, to protect against predators or for locomotion. This paper proposes a mathematical foundation to understand the intrinsic characteristics of the contact latch-based LaMSA mechanism. The authors analytically explain the question: ``What makes these small-scale systems impulsive?'' through a bifurcation that occurs by varying the input force on the latch and validate this proposition with numerical simulations. With the goal of controlling the latch mechanism, this paper explores the properties of the equilibrium points of the system. Finally, it introduces a sensitivity metric that provides insights into methods of control. The paper is structured as follows: Section 2 describes the generalized mathematical model of the contact latch-based system. Section 3 explores the existence of the fixed points in the system and gives the necessary and sufficient conditions for the saddle-node. Section 4 shows the existence of bifurcations in the contact latch-based LaMSA mechanism. Section 5 gives numerical simulations that validate the results on Matlab, followed by the conclusions in Section 6.
\section{Mathematical Model Description } \label{sec:model}
\begin{figure}[t!]
\includegraphics[width=8.5cm]{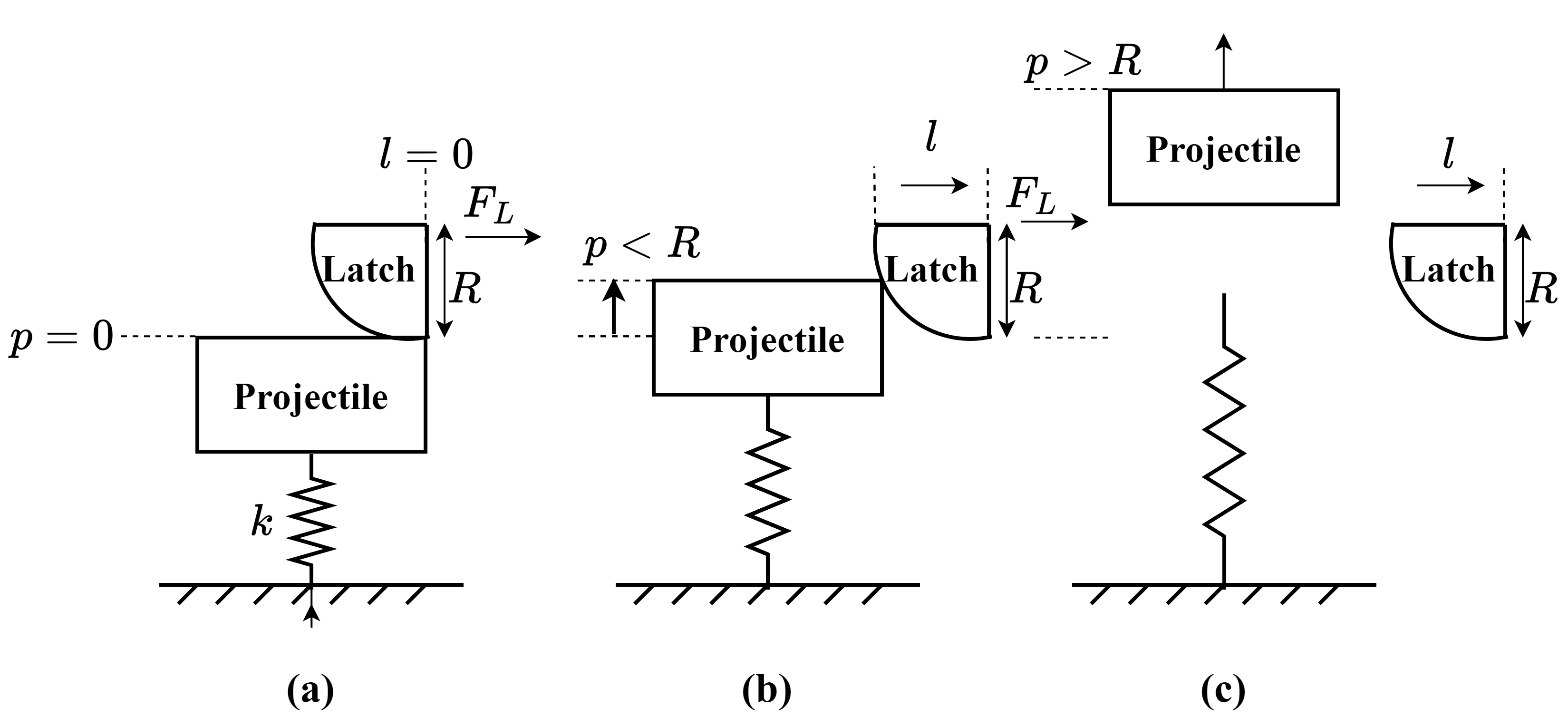}
\centering
\caption{The contact latch-based LaMSA model (a)~preloaded, (b)~in latched mode, (c)~unlatched: takeoff.}
\label{fig:contact_latch}
\end{figure}
The contact latch-based mechanism consists of two main components: a physical latch and a projectile held in place by a spring, as shown in Fig.~\ref{fig:contact_latch}.
Initially, the latch holds the projectile to store the energy of the spring (Fig.~\ref{fig:contact_latch}a). By applying the horizontal latch force, $F_L$, the latch releases the spring's potential energy and mediates the release (Fig.~\ref{fig:contact_latch}b). Once the latch is completely removed, the resultant potential energy is converted into the kinetic energy of the projectile, which provides an impulsive behavior (Fig.~\ref{fig:contact_latch}c). 

The notations of the spring position, $p$, latch position, $l$, are defined in Fig.~\ref{fig:contact_latch}, where $p_0$ indicates the spring's natural length, which is greater than the latch radius R.
The mathematical model of the generalized contact latch used in this paper is proposed by \cite{HO2023}, where the authors derive the dynamics based on constrained Lagrangian mechanics. We summarize the dynamic model here.

Let the state space be defined as $x:=(p,\dot{p}, l, \dot{l})$, where $p$ is the position of the projectile and $l$ is the position of the latch, and let $q=(p, l)$ and $\dot{q}=(\dot{p}, \dot{l})$. When the projectile and latch are physically constrained to each other, there exists a holonomic constraint $h(q)=0$, where
\begin{equation}\label{eq:hol}
  h(q)=l^2+(R-p)^2 -R^2.  
\end{equation}
Therefore, the system dynamics can be written as a switching system within two modes,
\begin{equation} \label{eq:model}
    \begin{bmatrix}
\dot{x}_1\\
 \dot{x}_2
\\ \dot{x}_3
\\ \dot{x}_4
\end{bmatrix}= \begin{bmatrix}
\dot{p}\\
\frac{1}{m}(F_s + (p-R)\lambda)\\ 
\dot{l}\\ 
\frac{1}{M}(F_L +l\lambda)
\end{bmatrix} :=f(x,F_L) 
\end{equation}
where,
\begin{eqnarray*}
\lambda &:=&\left\{\begin{matrix}
\tau(R,q,\dot{q},F_s,F_L)  \:\: \mbox{if}\: \:  h(q)=0 \: \: \mbox{and}\: \: \tau >0\: \: \mbox{(Latched)} \\
 0 \: \: \: \: \:\: \: \: \: \: \: \: \mbox{otherwise (Unlatched)}
\end{matrix}\right.
\\
F_s &:=& -k(p-p_0)
\end{eqnarray*}
such that $F_s$ is the Hookean spring force with a spring stiffness, $k$, and $m$ is the mass of the projectile, $M$ is the mass of the latch, $R$ is the latch radius, $F_L$ is the latch force, which is a control parameter, $\lambda$ is a contact force between the latch and the projectile, and $f(x, F_L)\in\mathbb{R}^{4}$ is the corresponding vector field. 
The contact force can either be zero (no force between the projectile and latch) or positive (both bodies are pushing against each other). 
We can derive the contact force, $\tau$, when the holonomic constraint is active using the constrained Lagrangian formulation in \cite{HO2023},
\begin{equation*}
    \tau = n_1(q,\dot{q},R) + g_1(q,R)F_s + g_2(q,R)F_L
\end{equation*}
where,
\begin{equation*}
\begin{split}
     n_1(q,\dot{q},R) &= -\left(\frac{(R-p)^2}{m}+\frac{l^2}{M}\right)^{(-1)} (\dot{p}^{2}+\dot{l}^{2}), \\
     g_1(q,R) &= -\left(\frac{(R-p)^2}{m}+\frac{l^2}{M}\right)^{(-1)}\frac{(p-R)}{m},\\
     g_2(q,R) &= -\left(\frac{(R-p)^2}{m}+\frac{l^2}{M}\right)^{(-1)}\frac{l}{m}.
\end{split}
\end{equation*}
Based on the activation of the constraints and the conditions such that contact force can be only positive (meaning that two bodies do not pull each other), we can define the Latched manifold as 
\begin{equation*}
    \mathrm{M}=\left\{ x \in \mathbb{R}^4 \:|\: h(q)=0 \wedge \tau>0 \right\}
\end{equation*}
and $\partial\mathrm{M}$ as the switching manifold.

Finally, there are two modes in the system, which we denoted as 
\begin{defin}[Latched Mode]
The system is in the latched mode when the projectile is in contact with the latch, and there exists non-negative force pushing against each other, namely $x\in\mathrm{M}$ 
\end{defin}
\begin{defin}[Unlatched Mode]
The system is in the unlatched mode when the projectile and latch are not pushing against each other, namely $x\notin\mathrm{M}$
\end{defin}
When the state $x$ hits the boundary of the latched manifold, $\partial \mathrm{M}$, the holonomic constraint in Eq.~(\ref{eq:hol}) no longer holds, switching the mode from a fully actuated 1 DoF system to an underactuated 2 DoF system. Therefore, the system dynamics is categorized as a multi-model multi-dimensional system $(\mbox{M}^3\mbox{D})$.

\textbf{Remark 1:} Along with the holonomic constraint, The model poses certain physical constraints on the projectile and latch position. When $p=R$, the constraint force does not affect the projectile. This behavior can be physically interpreted as when the projectile is tangential to the latch, no amount of physical force can hold the projectile in the latched position, leading to the takeoff condition. Thus, the projectile position is physically constrained within $0 \leq p \leq R$ to remain latched.  Following from Eq.~(\ref{eq:hol}), the latch is also constrained with $0 \leq l \leq R.$


\textbf{Remark 2:} This paper focuses on switching from the latched to unlatched mode, i.e., when the state $x$ crosses $\partial \mathrm{M}$. Currently, once the state cross the switching manifold, the projectile takes off with an impulse. In the future, this switching can be generalized to bring the projectile back to the latched mode by reloading the spring via external control, making the behaviour of the hybrid system more complex.

\textbf{Remark 3:} This paper uses the rounded latch design for the contact latch model. In \cite{DB2020}, multiple rounded latch designs were tested to provide various impulsive motions. In biology, different organisms have varied traits that can be classified under LaMSA. These traits can be studied using this generalized model by changing the latch type and other design parameters, such as $p_0$, $k$,$m$, etc.

\section{The Analysis of a saddle fixed point in Latched Mode} \label{sec:saddle}
This section derives the necessary and sufficient conditions of saddle fixed points in the Latched Mode of the system dynamics. Since the latch force, $F_L$ is the control parameter of the system, the equilibrium of the system will be a function of the latch force.
\subsection{Motivating Example}
A vital characteristic of the LaMSA mechanism is the ability of the latch to mediate the transfer of the spring energy. This is physically analyzed by varying the latch removal velocity. Quick removal of the latch leads to ultra-fast impulsive behavior, as seen in the mantis shrimp and other animal species. However, slower removal velocities lead to a loss of stored energy, making the projectile take off with a lower velocity.

In this paper, through the analysis of the fixed points, we conclude that the mediation of the latch is attributed to the movement of saddle nodes in the latched mode. In Proposition~\ref{prop:fixed}, we show that the system has two fixed points in the latched mode, one stationary at $p^*=0$ and another that moves based on the latch force $F_L$ exerted. By the results of Proposition~\ref{prop:FL}, the fixed points exist in the system when $F_L\leq 0$. Thus, by increasing the latch force, we see the fixed point moves closer to the stationary fixed point at $p^*=0$. In Lemma~\ref{lem:saddles}, we give the necessary and sufficient conditions for the fixed point to be a saddle-node. 

Section~\ref{sec:sims} presents a numerical simulation of the contact latch. We see the movement of the saddle fixed point with different $F_L$ values in Fig.~\ref{fig:fixedpoints}. It is observed that the saddle node pulls the trajectory towards its stable component and then launches the projectile till it unlatches and eventually takes off. The quiver arrows show the projectile trajectory change as the saddle moves for each $F_L$. Trajectories moving closer to the saddle regions are pulled closer to the fixed point, thus spending a portion of their total energy, resulting in reduced takeoff velocity. Therefore, the latch mediation can be attributed to this intrinsic property of the LaMSA system, i.e., through the moving saddle-node in its latched mode.
\subsection{Fixed points in the Latched Mode}
Suppose that the fixed point of the system in the Latch Mode is denoted as  $x^*=(p^*, \dot{p}^*, l^*, \dot{l}^*)$ with a constant $F^*_L$. The fixed points of the system are found from the state space described in Eq.~(\ref{eq:model}) by equating $f(x^*, F^*_L)=0$. The following proposition holds.

\vspace{12pt}
\begin{propos}
\label{prop:FL}
 If $x^*$ is a fixed point of the system with a constant latch force $F^*_L$ in Latched Mode, then $F_L^*$ must be non-positive. 
\end{propos}
\begin{proof}
Suppose that $F_L^*\in\mathbb{R}$ is fixed. We begin by deriving the equilibrium points $x^*=(p^*,\dot{p}^*,l^*,\dot{l}^*)$, which satisfy the condition in Eq.~(\ref{eq:model}): \[ \dot{p}^*=0, \:\: \dot{l}^*=0, \: \: \frac{1}{m}\left( F_s + (p^*-R) \lambda\right)=0,
\]
 Thus,
\begin{equation*}
\sqrt{R^2-(R-p^*)^2} (   F_s\frac{\sqrt{R^2-(R-p^*)^2}}{M} - \frac{p^*-R}{m}F_L^*)=0.
\end{equation*}
This implies that 
 \begin{equation}
     p^*=0,\:\: p^*=2R \:\: \mbox{or} \:\: F_s(\frac{\sqrt{R^2-(R-p^*)^2}}{M})=\frac{(p^*-R)}{m}F_L^* \label{eq:fp}
 \end{equation} where $F_s=-k(p^*-p_0)$ and $p^*< p_0$. As our analysis is restricted to the latched mode, we do not consider $p^*= 2R$ as a fixed point as it is outside our region of interest. Then,
 \[ F_s(\frac{\sqrt{R^2-(R-p^*)^2}}{M})\geq 0 \] Thus, for Eq.~(\ref{eq:fp}) to hold, we must have 
 \begin{equation}
 \label{eq:FL_nonpositive}
     F_L^* \leq 0.  
 \end{equation} 
 \end{proof}
 \begin{corol} \label{col:fp}
     When the latch force $F_L^*=0$ in the latched mode, the corresponding fixed point is $p^*=0$.
 \end{corol}
 \begin{proof}
     The fixed point has to satisfy Eq.~\ref{eq:fp}, thus when $F_L=0$,
     \[ F_s(\frac{\sqrt{R^2-(R-p^*)^2}}{M})= 0 \]
which implies that $p^*=0$, as $F_s >0$.     
 \end{proof}
 The result of Proposition~\ref{prop:FL} matches with our intuition that we must push the latch towards the projectile to balance the spring force (Fig~\ref{fig:contact_latch}a).
 \vspace{12pt}
 \begin{propos} \label{prop:fixed}
     In the Latched Mode, if $F_L^*<0$, then the system has two real fixed points, 
     \begin{equation}
    E_{F^*_L} := \{(0,0,0,0)^\top, (p^*, 0, l^*, 0)^\top\}\subset \mathrm{M}
\end{equation}
where $p^*$ $\in (0, R)$ and $l^*$ $\in (0, R)$ satisfies Eq.~(\ref{eq:hol}) and Eq.~(\ref{eq:fp}).
  \end{propos} 
 \begin{proof}
     
  From Eq.~(\ref{eq:fp}), we get the fixed point $p^*=0$, which is the first fixed point in the latched mode. In order to solve for $p^*$, we square Eq.~(\ref{eq:fp}) on both sides to get,
  \[{k^2(p_0^2-2p^*p_0+p^{*2})}\frac{2Rp^*-p^{*2}}{M^2} - \frac{p^{*2}+R^2-2Rp^*}{m^2}F_L^{*2}=0, \] which on expanding yields the fourth order equation,
\begin{equation} \label{eq:fixedpoint}
\begin{split}
    &D(p^*)=\frac{-k^2}{M^2}p^{*4}+(\frac{2Rk^2}{M^2}+\frac{2k^2p_0}{M^2})p^{*3}-(\frac{F_L^{*2}}{m^2}+\frac{k^2p_0^2}{M^2}\\
    &+\frac{4Rk^2p_0}{M^2})p^{*2} +(\frac{2F_L^{*2}R}{m^2}+\frac{2Rk^2p_0^2}{M^2})p^* - \frac{F_L^{*2}R^2}{m^2}=0.
\end{split}
\end{equation}
By squaring Eq.~(\ref{eq:fp}), the condition that $D(p^*)=0$ becomes a necessary condition for $p^*$ to be a fixed point of the system but is not sufficient, as $F_L^*$ may not always satisfy Proposition \ref{prop:FL}. Applying Descartes' Rule of Signs, we see that $D(p)$ can have either four real positive roots, a combination of two positive real roots with a complex conjugate, or two pairs of complex conjugate roots.

As $D(p^*)$ is a polynomial and thus is continuous in $\mathbb{R}$, we apply the Intermediate Value Theorem (IVT) in the interval $[0,2R]$ where,
\begin{equation*}
    \begin{split}
        &D(0)= \frac{-F_L^2R^2}{m^2}<0\\
        &D(R)=\frac{k^2R^2(R-p_0)^2}{m^2}>0\\
        &D(2R)= \frac{-F_L^2R^2}{m^2}<0\\
    \end{split}
\end{equation*}
Therefore, by IVT $\exists$  $p^* \in (0, R)$ such that $D(p^*)=0$ and $\exists$  $p^{*} \in (R, 2R)$ such that $D(p^{*})=0$.

Following the Complex Conjugate Root theorem, $D(p)$ has either (i) four real positive roots where at least one root $p^*_1\in (R,2R)$ and a maximum of three roots $\{p^*_2, p^*_3,p^*_4\}\subset (0,R)$  or (ii) a combination of two real positive roots and a complex conjugate pair with one root $p_4^*\in (0,R)$ and one root $p_5^*\in (R,2R)$. 

In the rest of the proof, we will prove that there must exist only one $p^*$ in $(0, R)$. We only need to check for case (i) to prove this. First, observe that the coefficient with the highest order in $D(p^*)$ is negative. If the stationary point, which makes $D^\prime(p^*)=0$ within $(0,R)$, exists at most once, then there exists only one $p^*$ in $(0,R)$ such that $D(p^*)=0$.

We check $D^\prime(p^*)$ evaluated at $p^*=0$ and $p^*=R$, 
\begin{equation*}
\begin{split}
D^{'}(p^*)=& \frac{-4k^2}{M^2}p^{*3}+(\frac{2Rk^2}{M^2}+\frac{2k^2p_0}{M^2})3p^{*2}-(\frac{F_L^2}{m^2}+\frac{k^2p_0^2}{M^2}\\
    &+\frac{4Rk^2p_0}{M^2})3p^{*} +(\frac{2F_l^2R}{m^2}+\frac{2Rk^2p_0^2}{M^2})
    \end{split}.
\end{equation*}
Applying IVT for $D^{'}(p^*)$ with
\begin{equation*}
    \begin{split}
        &D^{'}(0)= \frac{2F_l^2R}{m^2}+\frac{2Rk^2p_0^2}{M^2}>0\\
        &D^{'}(R)=\frac{2k^2R^2(R-p_0)}{m^2}<0\\
    \end{split}
\end{equation*}

we have at least one root of $D^\prime(p^*)=0$ in $(0,R)$. Furthermore, if there is at most one stationary point for $D^{\prime\prime}(p^*)$ in $(0, R)$, then there only exists one $p^*$ such that $D^\prime(p^*)=0$.


Therefore, by taking the second derivative, we get,
\begin{equation*}
\begin{split}
D^{''}(p^*)=& \frac{-12k^2}{M^2}p^{*2}+(\frac{2Rk^2}{M^2}+\frac{2k^2p_0}{M^2})6p^{*}-(\frac{F_L^2}{m^2}\\
&+\frac{k^2p_0^2}{M^2}+\frac{4Rk^2p_0}{M^2}) 
    \end{split}
\end{equation*}
Now observe that, 
\begin{equation*}
    \begin{split}
        &D^{''}(0)= -(\frac{F_L^2}{m^2}+\frac{k^2p_0^2}{M^2}+\frac{4Rk^2p_0}{M^2})<0\\
    \end{split}
\end{equation*}
Furthermore, the coefficient of the highest order is negative. Therefore, if there are no stationary point for $D^{\prime\prime\prime}(p^*)$ in (0,R), then $D^{\prime\prime}(p^*)=0$ will have at most one root. 

Now, by taking the third derivative and evaluating at $p^*=0$ and $R$, we have 
\begin{equation*}
    \begin{split}
    & D^{\prime\prime\prime}(p^*)=\frac{-24k^2}{M^2}p^{*}+\frac{12k^2(R-p_0)}{M^2}\\
        &D^{\prime\prime\prime}(0)= \frac{12k^2(R-p_0)}{M^2}<0\\
        &D^{\prime\prime\prime}(R)= -\frac{12k^2(R+p_0)}{M^2}<0\\
    \end{split},
\end{equation*}
since $p_0>R$, which implies that $D^{\prime\prime\prime}(p)<0$ for all $p\in(0,R)$.
Therefore, there is no stationary point for $D^{\prime\prime}(p^*)$ in $(0,R)$, which implies that there exists only one $p^*$ such that $D^\prime(p^*)=0$. Hence, there must only be one $p^*$ in $(0,R)$ satisfying Eq.~(\ref{eq:fp}).




    
\end{proof}

\subsection{Necessary and Sufficient conditions for a saddle-node in Latched Mode.}
To investigate the stability at a given fixed point, $p^*\in (0, R)$, and the $F^*_L$ satisfying Eq.~(\ref{eq:fp}), we linearize the nonlinear vector field $f$ defined in Eq.~(\ref{eq:model}) over $x$,  
\begin{equation} \label{eq:jacobian}
J=\frac{\partial f(x,F_L)}{\partial x}|_{(x,F_L)=(x^*, F_L^*) }=\begin{bmatrix}
0 & 1 & 0  &0  \\ 
\mathrm{A} & 0 &\mathrm{B} & 0\\
0& 0& 0& 1\\
\Gamma & 0 & \Delta& 0
\end{bmatrix}
\end{equation}
where $\mathrm{A}$, $\mathrm{B}$, $\Gamma$, $\Delta \in \mathbb{R}$ are defined as follows:

\begin{equation*}
\begin{split}
    &\mathrm{A}=\frac{-1}{m^2}(km+\frac{M(F_L^*l+2kp_0p^*-k(3p^{*2}+4pR-2p_0R+R^2)}{l^{*2}m+M(p^*-R)^2}\\
    &+\frac{2M^2(p^*-R)( R(-F_L^*l^*+kp_0R) +p^*(F_L^*l^*+k(2p_0-R)R) }{(l^{*2}m+M(p^*-R)^2)^2}\\
    &+\frac{ p^{*2}(k(p_0-2R))- p^{*3}k)}{(l^{*2}m+M(p^*-R)^2)^2}),
\end{split}
\end{equation*}
\[\mathrm{B}=\frac{F_L^*M(l^{*2}m-M(p^*-R)^2)(p^*-R)}{m(l^{*2}m+M(p^*-R)^2)^2)},\]
\begin{equation*}
    \begin{split}
        \Delta&=\frac{(F_L^*l^*(m-2l^{*2}m-2M(p^*-R)^2)}{(l^{*2}m+M(p^*-R)^2)^2}\\
        &+k(p_0-p^*)(-l^{*2}m+M(p^*-R)^2)(p^*-R))
    \end{split}
\end{equation*}
\normalsize
  and 
  \begin{equation*}
      \begin{split}
          \Gamma=&\frac{1}{(l^{*2}m+M(p^*-R)^2)^2}(l^*(2F_L^*l^*mM(p^*-R)\\
          &+k(2mM(-p^*+p_0)+mM(2p^*-p_0-R))(p^*-R)^2\\
          &+kl^{*2}m^2(2p^*-p_0-R)))
      \end{split}
  \end{equation*}
  \normalsize
Now, by constructing the characteristic equation for the linearized system, we have 
\[
|\lambda I-J|=\left|\begin{bmatrix}
\lambda & -1 & 0  &0  \\ 
-\mathrm{A} & \lambda &-\mathrm{B} & 0\\
0& 0& \lambda& -1\\
-\Gamma & 0 & -\Delta& \lambda
\end{bmatrix}\right|=0
\]
which results in \[\lambda(\lambda(\lambda^2-\Delta)+\mathrm{B}(0)) +1(-\mathrm{A}(\lambda^2-\Delta)+\mathrm{B}(-\Gamma)+0)=0.\]
Simplifying the expression leads to,
\begin{equation}\label{eq:eigen}
\lambda^4 -(\mathrm{A}+\Delta)\lambda^2 +(\mathrm{A}\Delta-\mathrm{B}\Gamma)=0.\\
\end{equation}
The following lemma shows the necessary and sufficient condition for the system to have a saddle fixed point. 
\\
\begin{lemm} \label{lem:saddles}
The fixed point $(p^*,F_L^*)$ of the system in the Latched Mode satisfying Eq.~(\ref{eq:FL_nonpositive}-\ref{eq:fixedpoint}) is a saddle if and only if  $(p^*,F_L^*) \in S$ where,
\[
S=\left\{(p^*,F_L^*)\:| h_1<0 \vee (h_1 \geq0 \wedge h_2>0) \right\}
\] 
and $h_1=( \mathrm{A}\Delta-\mathrm{B}\Gamma)$ and $h_2=(\mathrm{A}+\Delta)$.
\end{lemm}
\begin{proof}
Let $(p^*,F_L^*)$ be a fixed point of the system, then $(p^*,F_L^*)$ satisfies Eq.~(\ref{eq:fixedpoint}) $\forall F_L \leq 0$. 

$\Leftarrow$ Let $(p^*,F_L^*) \in S$. We can comment on the type of fixed point by using Descartes' Rule of Signs for the following cases:
\begin{enumerate}
    \item $h_1<0$ : The characteristic equation in Eq.~(\ref{eq:eigen}) takes the form \[
    \lambda^4-b\lambda^2-c=0 
    \]
    where $c>0$. In all cases when $b\geq0$ or $b<0$, the polynomial has one positive root and one negative root, as there is one sign change in each case, making it a saddle with a pair of eigenvalues with opposite sign.

    \item $h_1 >0 \wedge h_2>0$: The characteristic equation in Eq.~(\ref{eq:eigen}) takes the form \[
    \lambda^4-b\lambda^2+c=0 
    \]
    where $b>0$ and $c>0$. This polynomial has two positive and two negative roots, as there are two sign changes in each case, making it a saddle with two pairs of eigenvalues with opposite signs.
   
    \item $h_1 =0 \wedge h_2>0$: The characteristic equation in Eq.~(\ref{eq:eigen}) takes the form \[
    \lambda^4-b\lambda^2=0 
    \]
    where $b>0$. This polynomial has two roots at zero, one positive root, and one negative root, as there is one sign change in each case, making it a saddle with one pair of eigenvalues with opposite signs.
\end{enumerate}

$\Rightarrow$ Let the fixed point $(p^*,F_L^*)$ be a saddle. Comparing to Eq.~(\ref{eq:eigen}), the general structure of the characteristic equation for the saddle is $(\lambda^2-r_1)(\lambda^2-r_2)=0$ where $r_1,r_2$ are the roots, which is equivalent to 
\[
    \lambda^4-(r_1+r_2)\lambda^2+r_1r_2=0.
\]
Note that as the coefficients of $\lambda^3$ and $\lambda$ in Eq.~(\ref{eq:eigen}) is zero, we will not have saddle cases with real eigenvalues of different magnitude with opposing signs i.e., $ \lambda_1 < 0 < \lambda_2$ where $|\lambda_1|\neq |\lambda_2|$.
We comment on the saddle in the following cases:
\begin{enumerate}
     \item There is at least one pair of real eigenvalues with opposite signs. Here, $(r_1<0 \wedge r_2>0)\vee (r_1>0 \wedge r_2<0) \implies r_1\cdot r_2<0.$ Thus, $\mathrm{A}\Delta-\mathrm{B}\Gamma<0$.  
    
    \item  There are two pairs of real eigenvalues with opposing signs. Here, $(r_1>0 \wedge r_2>0)\vee(r_1<0 \wedge r_2<0) \implies r_1\cdot r_2>0 $. But, $r_1+r_2>0 \implies (\mathrm{A}\Delta-\mathrm{B}\Gamma>0)\wedge(\mathrm{A}+\Delta>0)$.


    
    \item There is one pair of real eigenvalues with opposing signs, and the other pair of eigenvalues is zero. Here, $ (r_1>0\wedge r_2=0)\wedge(r_1=0 \wedge r_2>0) \implies r_1\cdot r_2=0 $ and $ r_1+r_2>0 \implies (\mathrm{A}\Delta-\mathrm{B}\Gamma=0)\wedge(\mathrm{A}+\Delta>0).$ 
    

\end{enumerate}
   
\end{proof}
\textbf{Remark 4:}  The contact latch system has no stable nodes, as each fixed point would need four real negative eigenvalues to be a stable node. Descartes' Sign Rule shows that a fixed point with the characteristic equation in Eq.~(\ref{eq:eigen}) can have at most two real negative eigenvalues. 

\textbf{Remark 5:} When the system is in the unlatched mode $\lambda=0$, and thus, the system reduces to two decoupled ODEs as follows,
\begin{equation*}
    \begin{bmatrix}
\dot{x}_1\\
 \dot{x}_2
\\ \dot{x}_3
\\ \dot{x}_4
\end{bmatrix} = \begin{bmatrix}
\dot{p}\\
\frac{1}{m}F_s\\ 
\dot{l}\\ 
\frac{1}{M}F_L
\end{bmatrix} 
\end{equation*}
Solving $\dot{x}=0$, we see that 
\[\dot{p}=0\:\:,\:\: \dot{l}=0\:\:, F_s=0\:\:,\:\: F_L=0\:\:,\]
Thus, no saddle-node exists for the system. The spring is at its natural length $p_0$ with no external force acting on it; the latch and projectile are decoupled with no controlling force $F_L$ acting on the latch.
\section{Bifurcations in Contact Latch Model }\label{sec:bif}

In this section, a methodical analysis is performed on the model presented in Section \ref{sec:model}, emphasizing the fixed (equilibrium) points of the system to show that a significant change occurs in its quantitative behavior as we increase our control input $F_L$, i.e., a bifurcation. 
\\

\begin{defin}[Bifurcation]
A bifurcation is a qualitative change in the structure of the phase space of a system when a control parameter is varied such that
\begin{enumerate}
    \item a fixed point vanishes or is created,
    \item the stability of a fixed point changes.
\end{enumerate}
\end{defin}
It should be noted that the system behavior of interest is when the projectile and latch are in contact with one another i.e., when $l^2=R^2-(R-p)^2$. To further investigate the behavior of the fixed point, we look into finding a closed-form expression for $p^*$ from Eq.~(\ref{eq:fixedpoint}). 
However, getting a closed-form expression of $p^*$ from this expression proves to be very challenging. Our goal with this analysis is to find the trajectory of $p^*$ with a change in the latch force $F_L$. Thus, we can find the change in $p^*$ by implicitly differentiating the expression with respect to $F_L$. For brevity, we rewrite Eq.~(\ref{eq:fixedpoint}) as,
\[\mathrm{E}p^{*4}+\mathrm{Z}p^{*3}+\mathrm{H}p^{*2}+\Theta p^{*}+\mathrm{I}=0.\]
Differentiating implicitly with respect to $F_L$,
\[
\frac{dp^*}{dF_L}=\frac{-(\dot{\mathrm{E}}p^{*4}+\dot{\mathrm{Z}}p^{*3}+\dot{\mathrm{H}}p^{*2}+\dot{\Theta}p^{*}+\dot{\mathrm{I}})}{4\mathrm{E}p^{*3}+3\mathrm{Z}p^{*2}+2\mathrm{H}p^{}+\Theta}
\]
where, $\dot{\mathrm{E}}=0$, $\dot{\mathrm{Z}}=0$, $\dot{\mathrm{H}}=\frac{-2F_L}{m^2}$, $\dot{\Theta}=\frac{4RF_L}{m^2}$, $\dot{\mathrm{I}}=\frac{-2F_LR^2}{m^2}$. After substitution,
\begin{equation} \label{eq:nominal}
    \frac{dp^*}{dF_L}=\frac{\Lambda}{\Sigma+\Pi}
\end{equation}
where,
\begin{equation*}
    \begin{split}
        &\Lambda=-M^2(p^* - R)^2 F_L, \: \: \Sigma=F_L^2M^2(p^*-R),\\
        &\Pi=m^2(2k^4p^3+k^2p(p_0^2-3p(p_0+R))+(4R-p_0)p_0Rk^2).
    \end{split}
\end{equation*} 
Solving this ODE (for reference, we will call Eq.~(\ref{eq:nominal}) the nominal equation), with an initial condition $(p^*_0,F_{L0}^*)$ as an initial value problem (IVP), we can find the trajectory of the fixed point.
\\

\begin{propos}
\label{prop:bifurcation}
If the projectile $p$ starts at a fixed point $(p^*,F_{L}^*)$ in the Latched Mode, it will continue to move along with respect to Eq.~(\ref{eq:nominal}) until it reaches $(p^*,F_L^*)=(0,0)$ and vanishes after $F_L$ becomes positive, which indicates the bifurcation of the system.
\end{propos}
\begin{proof}
    Let $(p^*,F_{L}^*)$ be a fixed point to the system, thus $(p^*,F_{L}^*)$ must satisfy Eq.~(\ref{eq:fp}). As stated in Proposition~\ref{prop:fixed} and Corollary~\ref{col:fp}, the solution exist for all $F_L^*\leq 0$, implying the existence of solution of Eq.~(\ref{eq:nominal}) for the interval including $F^*_L\in [-\beta, 0]$ for any $\beta>0$. It is observed that as $F_L$ increases to zero, the corresponding $p^*$ will reach to $p^*=0$ when  $F_L^*=0$ (from Corollary \ref{col:fp}). As shown in Proposition~\ref{prop:FL} and Proposition~\ref{prop:fixed}, there does not exist a fixed point in the system when $F_L>0$. Hence, the system has the bifurcation in the Latched Mode.
\end{proof}

The bifurcation in the Latched Mode is shown in Proposition~\ref{prop:bifurcation}. It is not necessarily showing that the fixed point $p^*$ will be the saddle node across any $F_L<0$. The following corollary shows the sufficient condition for the bifurcation where the saddle node disappears.

Let $U$ be an open set such that the denominator of the nominal equation, Eq.~(\ref{eq:nominal}) is upper bounded,
$U \subset \mathbb{R}^2$ such that \[
     U=\left\{ (p^*,F_L^*) \in \mathbb{R}^2 \:\:|\:\: \Sigma+\Pi<0\right\}
    \]
    \\
\begin{corol}
    By choosing the LaMSA system design parameters $m, M, R$ and $p
    _0$ properly, there exist an open neighborhood of $(p^*, F_L^*)=(0,0)$, denoted as $V$, such that $V\subset U\cap S$, and the Latched system contains a bifurcation of the saddle node. 
\end{corol}
\begin{proof}
    First, observe that if we choose $p_0\in (R, 4R)$ then $(p^*, F_L^*)=(0,0)\in U$ since 
\begin{equation*}
        \Sigma+\Pi = m^2(4R-p_0)p_0Rk^2<0.    \end{equation*}
Now, by computing for $h_1$ in Lemma~\ref{lem:saddles} at $(0,0)$, we have 
\begin{equation}
    h_1 = \frac{k^2p_0(m+MR)}{m(m+M)^2R^2}\cdot(1+\frac{M(2(1-MR^2)p_0-R)}{m(m+M)^2R^3})
\end{equation}
and so to ensure $h_1<0$, it is sufficient to have big enough $R$ and choose the natural lenght $p_0$, to satisfy
\begin{equation}
    p_0>\frac{m(m+M)^2R^2-R}{2M(MR^2-1)}
\end{equation}
and $p_0\in (R, 4R)$. The proof of the existence of such design choice satisfying the above inequality can be done constructive, and we give an example in Sec.~\ref{sec:sims}. This implies that $(0,0)\in S$. Now, by Proposition~\ref{prop:bifurcation}, there exist a bifurcation near $(0,0)$ such that the saddle node disappears. 
\end{proof}

\textbf{Remark 6:} The proposition and the corollary above indicate that the LaMSA system can have the saddle and the bifurcation at the same time by properly choosing the system design. This explains why the LaMSA system can contain two fundamental features: the mediation of the potential energy release through the saddle node, and the impulsive behavior after the saddle disappears.

\begin{figure*}[h] 
\centering 
\subfloat[$F_L=-15\: \mbox{N}$]{\includegraphics[width=0.33\textwidth]{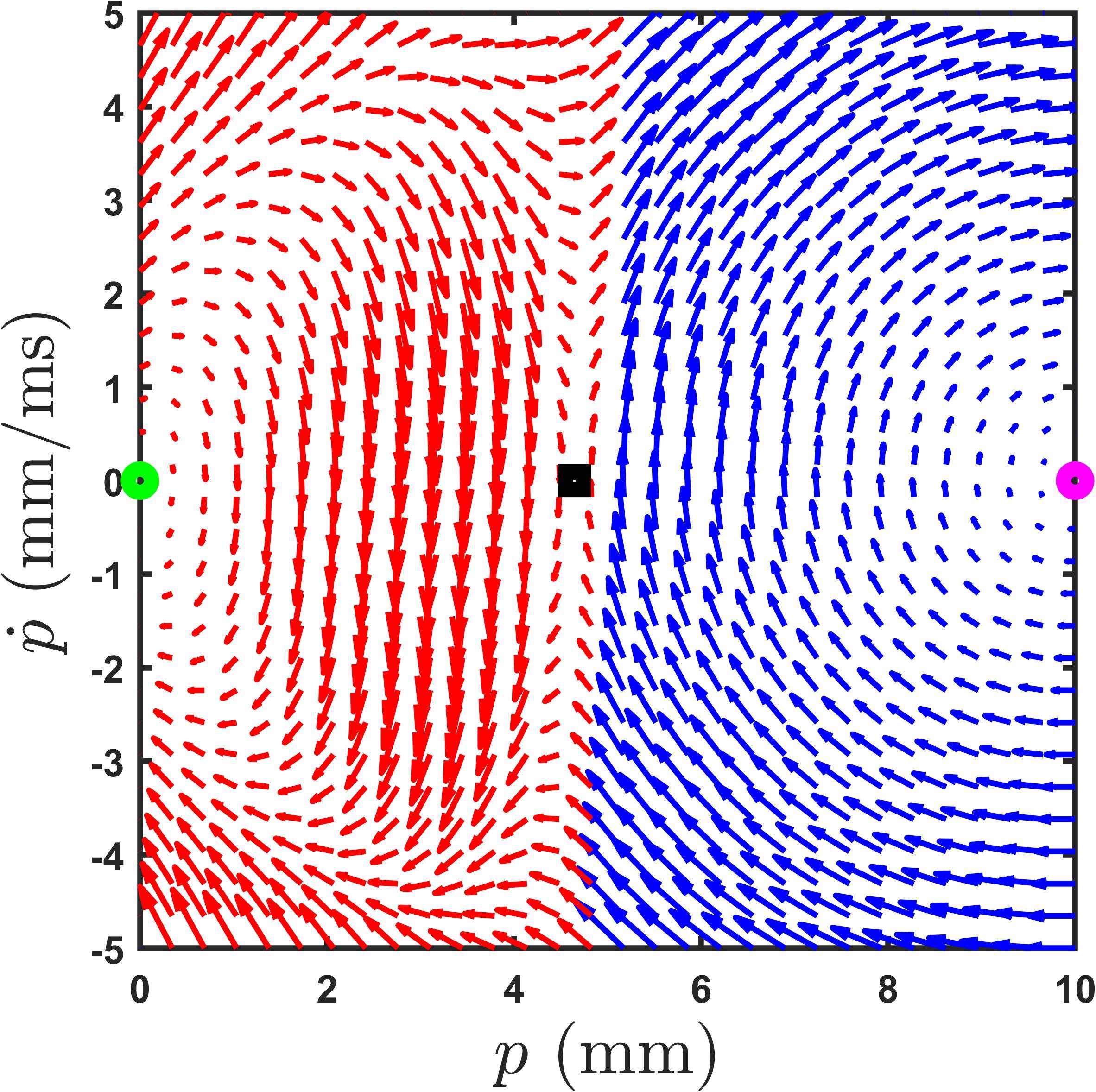}\label{fig:Fp_FL_15}}
\subfloat[$F_L=-5\: \mbox{N}$]{\includegraphics[width=0.33\textwidth]{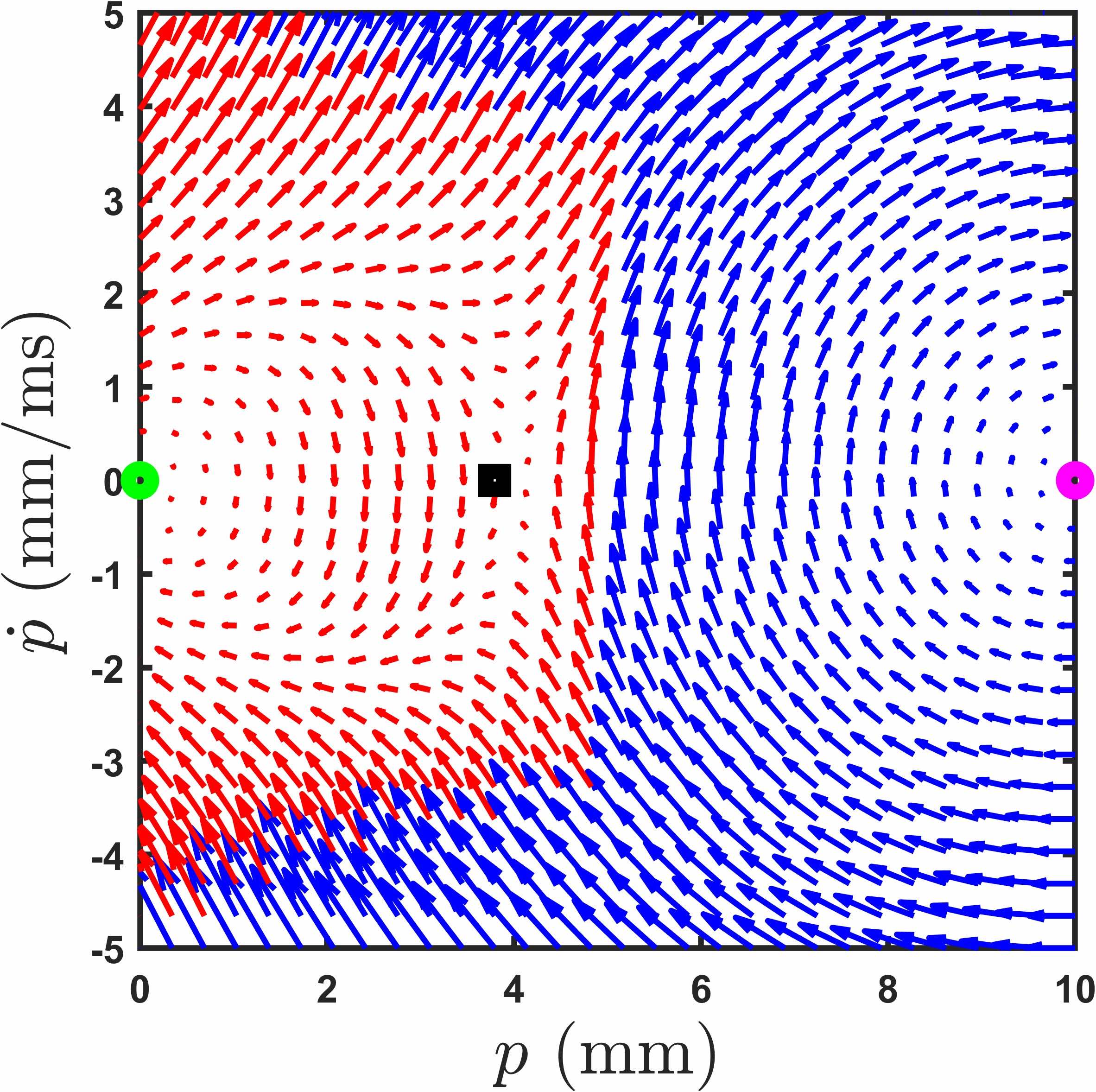}\label{fig:Fp_FL_10}}
\subfloat[$F_L=-1\: \mbox{N}$]{\includegraphics[width=0.33\textwidth]{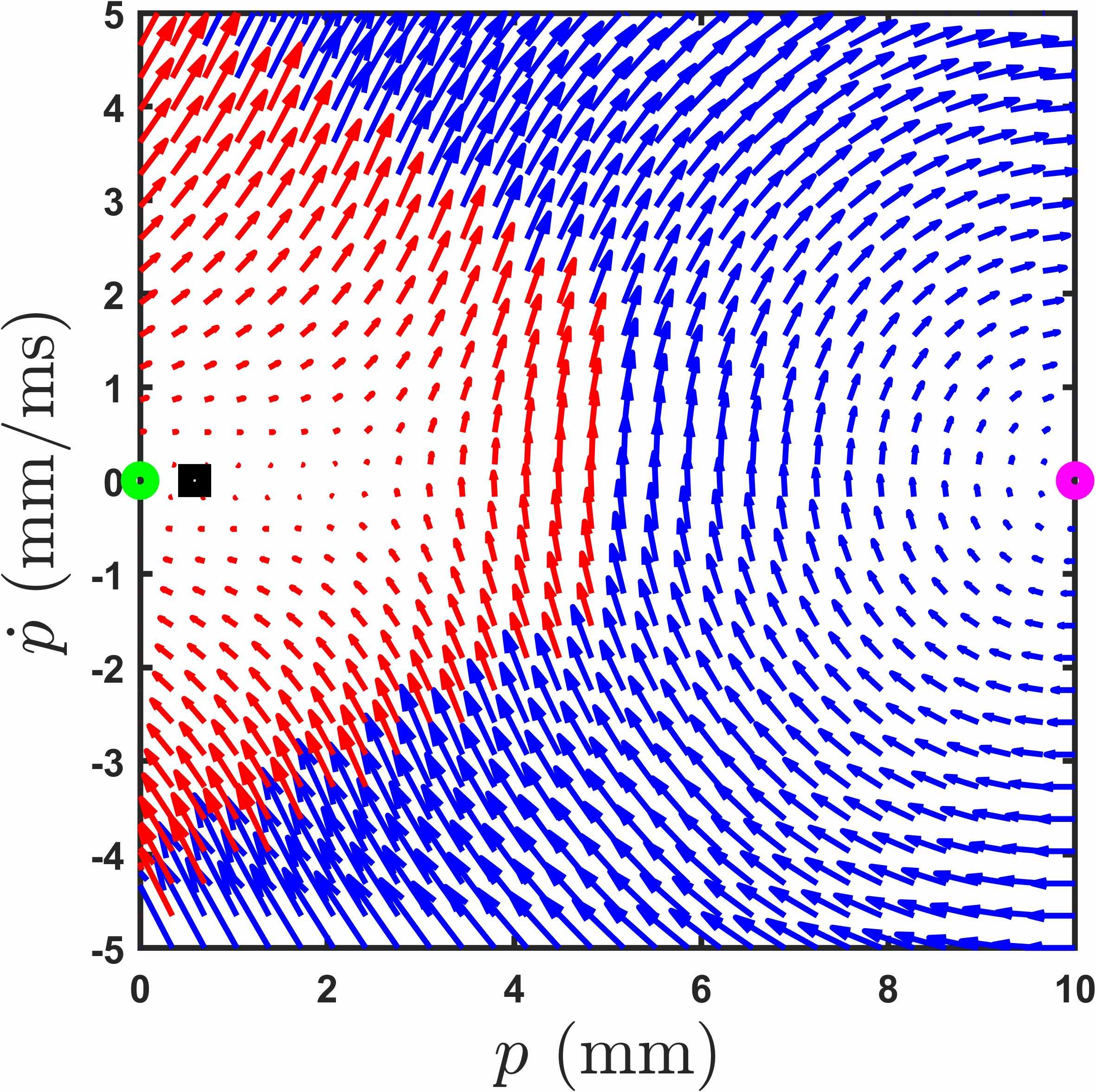}\label{fig:Fp_FL_1}/}
\caption{Fixed points in the Contact latch LaMSA system in Latched (Red) and Unlatched (Blue) mode. The stationary saddle is the green circle, the moving fixed point in Latched mode is the black square, and the fixed point in the Unlatched mode is the pink circle. }
\label{fig:fixedpoints}
\end{figure*}

\begin{figure}[ht!] 
\begin{center}
\includegraphics[scale=0.07]{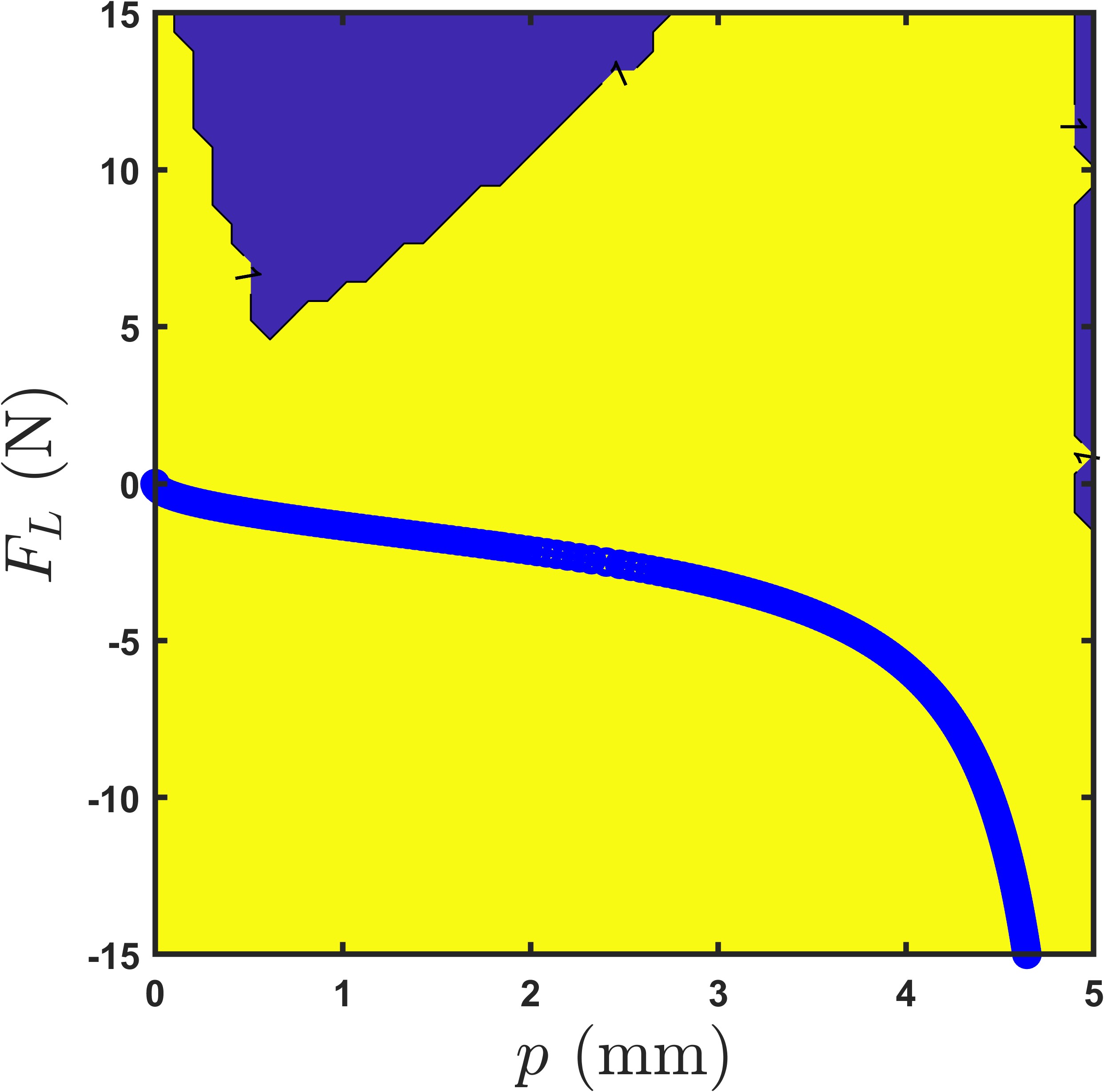}\label{fig:normal_15}
\caption{Trajectory of the nominal equation with varying initial conditions $(p^*,F_L^*)$.  The yellow region satisfies the saddle conditions, and the purple region does not satisfy the saddle conditions.} \label{fig:normal}
\end{center}
\end{figure}

\section{Numerical Simulations} \label{sec:sims}
In this section, we perform numerical simulations to validate the results presented in Section \ref{sec:bif}. The following parameters are used: $m=1$ mg, $M=5$ mg, $R=5$ mm, $k=1$ N$\mbox{m}^{-1}$ and $p_0=10$ mm. The latch force $F_L$ varies from $-15$ N to $0$ N. Intuitively, we can understand the negative $F_L$ as a pushing force on the latch. Solving Eq.~(\ref{eq:fixedpoint}) on Matlab, the fixed points of the system for each $F_L$ are plotted on the the phase portrait of projectile position $p$ vs velocity $\dot{p}$ (Fig.~\ref{fig:fixedpoints}). The phase portrait of the states $(p,\dot{p})$ for $F_L=-15$ N (Fig.~\ref{fig:Fp_FL_15}), $F_L=-10$ N (Fig.~\ref{fig:Fp_FL_10}) and $F_L=-1$ N (Fig.~\ref{fig:Fp_FL_1}) are shown in Fig.~\ref{fig:fixedpoints}. As per Proposition~\ref{prop:fixed}, the system has two real fixed points in the latched mode: one is stationary at $p=0$ (green circle), while the other fixed point varies its position based on $F_L$ (black square). We also observe the fixed point $p^*=p_0$ (pink circle) in the unlatched mode. Our analysis focuses on Latched Mode, which is the red sections. It is observed that as $F_L$ increases, the fixed point moves closer to the stationary point $p=0$ as expected.

\begin{figure}[t!]
\includegraphics[scale=0.07]{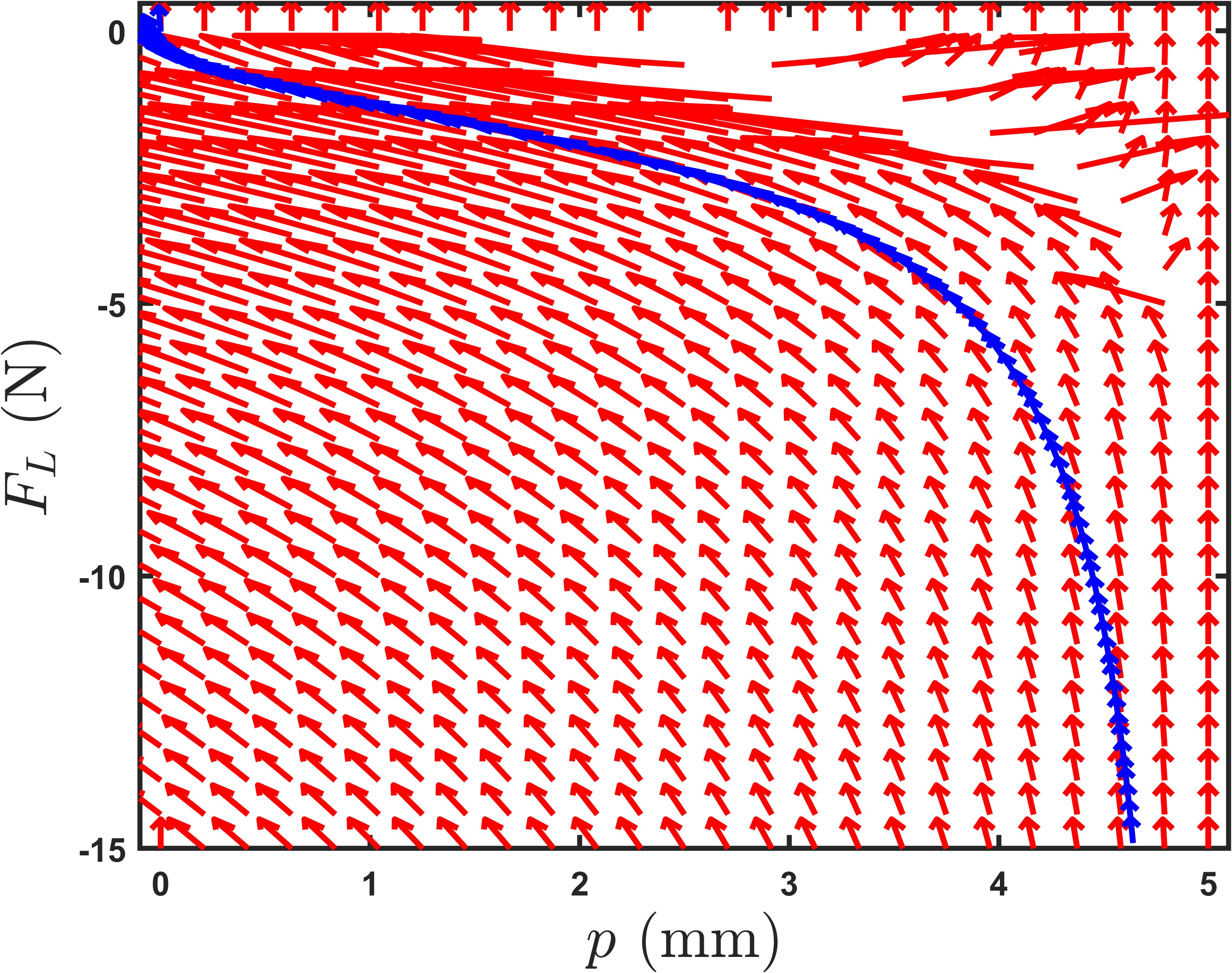}
\centering
\caption{Quiver plot of the nominal equation with the trajectory of saddle points $(p^*,F_L^*)$.}
\label{fig:quiver_normal}
\end{figure}

Numerically calculating the characteristic equation, Eq.~(\ref{eq:eigen}), for each $(p^*,F_L^*)$ pair shows that all the real fixed points of the system are saddles. We plot the region of the saddles $S$ based on Lemma~\ref{lem:saddles} (seen in Fig.~\ref{fig:normal}). The yellow region represents the set $S$, which satisfies the saddle conditions from Lemma~\ref{lem:saddles}. The purple region represents the complement of $S$ that does not satisfy the saddle condition. The movement of the fixed points in the Latched Mode is shown by solving Eq.~(\ref{eq:nominal}) numerically using ODE23 on Matlab. The initial conditions $(p^*,F_L^*)$ are fixed point of the system calculated from Eq.~(\ref{eq:fixedpoint}). 
It is observed that once the trajectory starts at a saddle, it continues along saddles until it reaches the stationary point (0,0). The saddle vanishes as $F_L$ becomes positive, which is a bifurcation. 

The phase portrait of the nominal equation, Eq.(\ref{eq:nominal}), is shown in Fig.~\ref{fig:quiver_normal} to show the sensitivity of the saddle point with the change of latch force, $F_L$. 
The red arrow indicates the flow of the nominal equation, and the blue arrow curve is the solution of Eq.(\ref{eq:nominal}) starting from the saddle pair $(p^*,F_L^*)=(4.64383, -15)$ computed in Fig.~\ref{fig:fixedpoints}.
Two interesting observations are noted: (i) No stable region is seen in the phase portrait, the trajectory at all points is pushed outward, and (ii) the sensitivity of the flow with respect to $F_L$ becomes minimal around the boundary conditions $p=R$, and $F_L=0$
. 
Thus, starting at the saddle trajectory, following the blue curve, the projectile will reach to the minimally sensitive fixed point, $(0,0)$, and vanish as $F_L$ becomes positive. 
This system behavior explains the impulsive nature of the LaMSA system concerning the change of latch control force $F_L$. 
As we increase the latch force to a pulling force (i.e., $F_L>0$), the system has no real fixed point other than $p=0$ and $p=R$. We know that when $p=R$, the system will unlatch, as discussed in Remark 1. Therefore, no positive force would keep the system at the saddle, as the projectile will move, causing the impulsive jump.    



\section{Conclusions}
In this paper, the authors present a foundational analysis to study the intrinsic characteristics of the contact latch-based LaMSA mechanism. The LaMSA systems form a class of hybrid multi-dimensional Multi-modal systems where the system state space switches from a 1 DoF system when latched to a 2 DoF system when unlatched. These impulsive movements are prevalent in various species of animals like the mantis shrimp, click beetle, grasshopper, and trap-jaw-ants. These animals all use this impulse for different purposes like locomotion or protection. Significant work has been done in robotics to mimic these movements precisely from a physiological perspective. 

Through mathematical analysis of the fixed points, the authors show how the mediation of the latch occurs via the movement of system fixed points. Through lemmas and propositions in Section \ref{sec:saddle}, the authors show that the contact latch system has exactly two fixed points in the latched mode, one stationary at $p=0$ and one fixed point that moves by varying the latch force $F_L$. By performing numerical simulations, the authors further conclude that the movement of the fixed points aids in mediating the energy transfer of the projectile and takeoff.

The concept of bifurcations has been used extensively in analyzing the stability of hybrid systems. In Section \ref{sec:bif}, the authors investigate the movement of the fixed point in the latched mode and show that the fixed point moves along the phase plane and eventually vanishes at the stationary fixed point $(p^*,F_L^*)=(0,0)$ causing a bifurcation in the system. As shown in Proposition~\ref{prop:FL}, the moving fixed point disappears once the latch force $F_L>0$. Thus, the system takes off with the impulsive movement once a pulling latch force is applied. During the analysis of the fixed point, a sensitivity in the movement of the fixed point is observed. As the saddle fixed point moves closer to $p^*=2$, 
a slight increase in the latch force causes a significant change in the position of $p^*$. This metric could be used while deciding a control strategy for various applications, such as making the projectile take off at a certain velocity.

\subsection{Future Work and Applications}
This paper introduces the framework for analyzing the behavior of the class of LaMSA system, starting with the contact latch model. Extending the analysis to other classifications of the LaMSA mechanism remains as a future work. It is interesting to see the effects of the system parameters, such as the spring length and latch radius, on the performance of the contact latch. Different latch types can be used in the model to mimic the behavior of different species, as mentioned in Remark 2. In the future, an optimal latch could be designed to enhance the projectile takeoff and to ensure maximal energy transfer in the spring.

Taking inspiration from biology, these impulsive systems can be used in a wide range of applications. \cite{HP2016} constructed a jumping robot that could jump up to 78 \% the vertical jumping agility of a galago and explore terrains of locomotion that previously couldn't be attained with traditional jumping robots. \cite{HX2022} created a framework to compare the energies of biological vs engineered jumpers to achieve the highest jump possible. Using this framework they designed a device that jumps over 30 meters in height.

\begin{ack}
The authors acknowledge Prof. Robert J. Wood for the discussion on conceptualizing the generalized LaMSA framework within the bifurcation framework.
\end{ack}
\bibliography{ifacconf}    

\end{document}